\documentclass[preprint,a4paper,11pt,sort]{elsarticle}
\usepackage{fullpage}

\usepackage{tikz}

\usepackage{amsmath,amssymb,latexsym}
\usepackage{amsthm} % we need this. https://stackoverflow.com/questions/6499504/cleveref-fails-for-theorem-environments-sharing-the-same-counter
\usepackage{thm-restate}

\usepackage{tcolorbox}

\usepackage{xcolor}
\definecolor{darkred}{rgb}{0.8,0,0}

\usepackage{graphicx}

\usepackage{thmtools} % workaround for a cleveref issue in TeX Live 2025: https://www.reddit.com/r/LaTeX/comments/1r979z6/comment/o6au2h3/?utm_source=share&utm_medium=web3x&utm_name=web3xcss&utm_term=1&utm_content=share_button

\usepackage[colorlinks=true]{hyperref}

\usepackage[capitalize,nameinlink]{cleveref}
\newtheorem{theorem}{Theorem}[section]

\newtheorem{lemma}[theorem]{Lemma}
\crefname{lemma}{Lemma}{Lemmas}

\newtheorem{claim}{Claim}[theorem]
\crefname{claim}{Claim}{Claims}

\crefname{corollary}{Corollary}{Corollaries}

\crefname{proposition}{Proposition}{Propositions}

\theoremstyle{definition}
\newtheorem{observation}[theorem]{Observation}
\crefname{observation}{Observation}{Observations}

\newenvironment{subproof}[1][\proofname]{%
  \begin{proof}[#1]%
}{%
  \end{proof}%
}

\newcommand{\STCfull}{\textsc{Spanning Tree Congestion}}
\newcommand{\STC}{\textsc{STC}}
\newcommand{\stc}{\mathrm{stc}} % spanning tree congestion
\newcommand{\cng}{\mathrm{cng}} % congestion
\newcommand{\dist}{\mathrm{dist}} % distance

\usepackage{mathtools}

\begin{document}

\title{Spanning tree congestion of proper interval graphs\tnoteref{t1}}
\tnotetext[t1]{%
  Partially supported by JSPS KAKENHI Grant Numbers
  JP22H00513, % Yota (Ono A)
  JP24H00697, % Yota (Tamaki A)
  JP25K03076, % Yota (Otachi B)
  JP25K03077. % Yota (Hanaka B)
} 

\author{Yota Otachi} %\author[nagoya-u]{Yota Otachi}
\ead{otachi@nagoya-u.jp}

\affiliation{ %\affiliation[nagoya-u]{
    organization={Nagoya University},
    city={Nagoya},
    country={Japan}
}

\begin{abstract}
We show that the spanning tree congestion problem is NP-complete even on proper interval graphs with linear clique-width at most~$4$ and diameter~$3$.
By slightly modifying the reduction, we also show that the problem is NP-complete on (general) graphs of diameter~$2$.

\par % to add line numbers in abstract
\end{abstract}


\maketitle

%%%%%%%%%%%%%%%%%%%%%%%%%%%%%%%%%%%%%%%%%%%%%%%%%%%%%%%%%%%%%%%%%%%%%%%%%%%%%%%%%%%%%%%%%%%%%%%%%%%%%%%%%%%%%%%%%%%%%%%%%%%%%%%%
%%%%%%%%%%%%%%%%%%%%%%%%%%%%%%%%%%%%%%%%%%%%%%%%%%%%%%%%%%%%%%%%%%%%%%%%%%%%%%%%%%%%%%%%%%%%%%%%%%%%%%%%%%%%%%%%%%%%%%%%%%%%%%%%
%%%%%%%%%%%%%%%%%%%%%%%%%%%%%%%%%%%%%%%%%%%%%%%%%%%%%%%%%%%%%%%%%%%%%%%%%%%%%%%%%%%%%%%%%%%%%%%%%%%%%%%%%%%%%%%%%%%%%%%%%%%%%%%%

\section{Introduction}

Let $G$ be a connected graph and $T$ be a spanning tree of $G$.
The \emph{congestion} of an edge $e \in E(T)$ is the size of the fundamental cut of $e$ with respect to $T$,
and the \emph{congestion} of $T$ is the maximum congestion of an edge of $T$.
The \emph{spanning tree congestion} of $G$, denoted $\stc(G)$, is the minimum congestion over all spanning trees of $G$.
The problem studied in this paper, \STCfull, is formalized as follows.
%

%\begin{tcolorbox}
\begin{description}
  \setlength{\itemsep}{0pt}
  \item[Problem.] {\STCfull} ({\STC})
  \item[Input.] A connected graph $G$ and an integer $k$.
  \item[Question.] Is $\stc(G) \le k$?
\end{description}
%\end{tcolorbox}

Our goal in this paper is to prove the following theorem.
\begin{restatable}{theorem}{thmPint}
\label{thm:stc-pint}
{\STCfull} is $\mathrm{NP}$-complete on proper interval graphs of linear clique-width at most~$4$ and diameter~$3$.
\end{restatable}

By slightly modifying the reduction used to prove \cref{thm:stc-pint},
we also show that {\STCfull} is $\mathrm{NP}$-complete on graphs of diameter~$2$.
\begin{restatable}{theorem}{thmDiamTwo}
\label{thm:diam2}
\STCfull{} is $\mathrm{NP}$-complete on graphs of diameter~$2$.
\end{restatable}

Note that \cref{thm:diam2} is not about proper interval graphs.
Indeed, every proper interval graph of diameter at most~$2$ is a cochain graph (see \cref{sec:cochain}),
for which a polynomial-time algorithm is known~\cite{KuboYKY14la,Kubo15mt}.

\subsection{Related results}
Since Ostrovskii~\cite{Ostrovskii04} formally introduced the concept of spanning tree congestion,
it has been studied extensively from both graph-theoretic and algorithmic perspectives.
Here we summarize relevant known results on the complexity of {\STC} on restricted graph classes.
For more information about spanning tree congestion, see the survey by Otachi~\cite{Otachi20}.

In general, {\STC} is known to be NP-complete~\cite[Section~5.6]{Lowenstein10}, and its NP-completeness holds 
even on planar graphs~\cite{BodlaenderFGOL12}, split graphs~\cite{OkamotoOUU11}, chain graphs ($2K_{2}$-free bipartite graphs)~\cite{OkamotoOUU11},
and graphs of maximum degree~$3$~\cite{AtaligCDKLSZ26}.
We recently noticed that a proof by Lampis et al.~\cite{LampisMNOVV25} can be easily modified to prove the NP-completeness of {\STC} on interval graphs of linear clique-width at most~$3$
(see the updated full version of their paper~\cite[Theorem~9]{LampisMNOVV25arxiv}).
The main result of this paper (\cref{thm:stc-pint}) strengthens this hardness to proper interval graphs.

On the other hand, it is known that {\STC} is polynomial-time solvable for trivially perfect graphs (chordal cographs)~\cite{OkamotoOUU11},
cochain graphs (cobipartite chordal graphs)~\cite{KuboYKY14la,Kubo15mt}, 
and outerplanar graphs~\cite{BodlaenderKMO11}.
Furthermore, the XP algorithm parameterized by treewidth~\cite[Theorem~14]{LampisMNOVV25}, 
together with the fact that every $p$-outerplanar graph has treewidth at most $3p-1$~\cite{Bodlaender98}, 
implies that {\STC} is polynomial-time solvable on $p$-outerplanar graphs for every fixed $p$.

See \cref{fig:graph-classes} for a summary of the results mentioned above.

\begin{figure}[tb]
  \centering
  % -*- coding:utf-8 -*-
%#BIBTEX bibtex main
%#!pdflatex -synctex=1 main

%\definecolor{tkzdarkred}{rgb}{.5,.1,.1}
%\definecolor{tkzdarkorange}{rgb}{.7,.25,0}
%\definecolor{tkzdarkblue}{rgb}{.1,.1,.5}
%\definecolor{tkzdarkgreen}{rgb}{.1,.35,.15}

%\tikzset{npc/.style = {draw,semithick,rectangle,tkzdarkred,fill=tkzdarkred!10,align=center}}
%\tikzset{ply/.style = {draw,semithick,rectangle,rounded corners,tkzdarkgreen,fill=tkzdarkgreen!12,align=center}}
\tikzset{npc/.style = {draw,semithick,rectangle,fill=red!20,align=center}}
\tikzset{ply/.style = {draw,semithick,rectangle,rounded corners,fill=green!20,align=center}}
\tikzset{unk/.style = {draw,semithick,rectangle,dotted,black,align=center,fill=gray!20}}

\begin{tikzpicture}[semithick, every node/.style={font=\small},scale=1.0]
  \def\colp{-3.2}
  \def\colc{5}
  \def\row{1.7}

  \node[npc] (gen) at (0, 5*\row) {general~\cite{Lowenstein10}};

  \node[npc] (deg8) at (-6, 4*\row) {$\Delta \le 8$~\cite{LampisMNOVV25}};
  \draw (gen) -- (deg8);

  \node[npc] (deg3) at (-6, 3*\row) {$\Delta \le 3$~\cite{AtaligCDKLSZ26}};
  \draw (deg8) -- (deg3);

  \node[npc] (pln) at (\colp, 4*\row) {planar~\cite{BodlaenderFGOL12}};
  \draw (gen) -- (pln);

  \node[ply] (pop) at (\colp, 3*\row) {$p$-outerplanar~\cite{Bodlaender98,LampisMNOVV25} \\ (XP w.r.t.~$p$) };
  \draw (pln) -- (pop);

  \node[ply] (opl) at (\colp, 2*\row) {outerplanar~\cite{BodlaenderKMO11}};
  \draw (pop) -- (opl);

  \node[npc] (chd) at (0, 4*\row) {chordal};
  \draw (gen) -- (chd);

  \node[npc] (int) at (0, 3*\row) {interval~\cite{LampisMNOVV25arxiv}};
  \draw (chd) -- (int);

  \node[npc, ultra thick] (pit) at (0, 2*\row) {proper interval \\ (\cref{thm:stc-pint})};
  \draw (int) -- (pit);

  \node[ply] (cch) at (0, 1*\row) {cochain~\cite{KuboYKY14la,Kubo15mt}};
  \draw (pit) -- (cch);

  \node[npc] (spl) at (1.9, 3.5*\row) {split~\cite{OkamotoOUU11}};
  \draw (chd) [out=0,in=120] to (spl);

  \node[npc] (cw3) at (\colc, 4*\row) {$\textrm{clique-width} \le 3$};
  \draw (gen) -- (cw3);

  \node[unk] (cog) at (\colc, 3*\row) {cograph \\ ($\textrm{clique-width} \le 2$)};
  \draw (cw3) -- (cog);

  \node[ply] (trp) at (\colc, 2*\row) {trivially perfect~\cite{OkamotoOUU11}};
  \draw (cog) -- (trp);
  \draw (int) -- (trp);

  \node[ply] (thd) at (\colc, 1*\row) {threshold};
  \draw (trp) -- (thd);
  \draw (spl) [out=280,in=170] to (thd);

  \node[npc] (chn) at (8, 1*\row) {chain~\cite{OkamotoOUU11}};
  \draw (cw3) [out=345,in=90] to (chn);

  \draw (cch) [out=0,in=195] to (cw3);
\end{tikzpicture}
  \caption{The complexity of {\STC} on graph classes.
    The (red) solid rectangles represent NP-complete cases and 
    the (green) rounded rectangles represent polynomial-time solvable cases.}
  \label{fig:graph-classes}
\end{figure}

\subsection{On a recently published algorithm for interval graphs}
\label{sec:lin-lin}

In 2025, Lin and Lin~\cite{LinL25} presented a polynomial-time algorithm for {\STC} on interval graphs.
However, as mentioned above, the updated full version of Lampis et al.~\cite[Theorem~9]{LampisMNOVV25arxiv} proves that the problem is NP-complete on interval graphs.
Unless $\mathrm{P} = \mathrm{NP}$, the claimed algorithm and the hardness result are incompatible.
After examining the paper of Lin and Lin, we found gaps in the description that prevent us from verifying the algorithm in its current form.
A subsequent private communication with the authors indicated that, in light of the NP-hardness result, they now believe their algorithm to be incorrect~\cite{LinL25p}.

%%%%%%%%%%%%%%%%%%%%%%%%%%%%%%%%%%%%%%%%%%%%%%%%%%%%%%%%%%%%%%%%%%%%%%%%%%%%%%%%%%%%%%%%%%%%%%%%%%%%%%%%%%%%%%%%%%%%%%%%%%%%%%%%
%%%%%%%%%%%%%%%%%%%%%%%%%%%%%%%%%%%%%%%%%%%%%%%%%%%%%%%%%%%%%%%%%%%%%%%%%%%%%%%%%%%%%%%%%%%%%%%%%%%%%%%%%%%%%%%%%%%%%%%%%%%%%%%%
%%%%%%%%%%%%%%%%%%%%%%%%%%%%%%%%%%%%%%%%%%%%%%%%%%%%%%%%%%%%%%%%%%%%%%%%%%%%%%%%%%%%%%%%%%%%%%%%%%%%%%%%%%%%%%%%%%%%%%%%%%%%%%%%

\section{Preliminaries}

For integers $m$ and $n$, let $[m,n]$ be the set of integers from $m$ to $n$,
and let $[n]$ be the set of positive integers less than or equal to $n$; that is,
$[m, n] = \{d \in \mathbb{Z} \mid m \le d \le n\}$
and $[n] = [1,n]$.

Let $G = (V,E)$ be a graph. 
We sometimes refer to the vertex set and the edge set of $G$ as $V(G)$ and $E(G)$, respectively.
For $S \subseteq V$, let $G[S]$ denote the subgraph of $G$ induced by $S$.
If $G[S]$ has all possible edges, then $S$ is a \emph{clique}.
For $v \in V$, we denote by $N_{G}(v)$ and $N_{G}[v]$ the (\emph{open}) \emph{neighborhood} and the \emph{closed neighborhood} of $v$ in $G$, respectively;
i.e., $N_{G}(v) = \{u \mid \{u,v\} \in E\}$ and $N_{G}[v] = N_{G}(v) \cup \{v\}$.
The degree of $v$ in $G$, i.e., $|N_{G}(v)|$, is denoted by $\deg_{G}(v)$.
We denote the minimum degree of $G$ by $\delta(G)$ and the maximum degree by $\Delta(G)$.
For two disjoint vertex subsets $A, B \subseteq V$, 
let $E(A,B)$ be the set of edges between $A$ and $B$; that is, $E(A, B) = \{\{a,b\} \in E \mid a \in A, \; b \in B\}$.

Now let $G = (V,E)$ be a connected graph and $T$ be a spanning tree of $G$.
For $e \in E(T)$, we define its \emph{congestion}, denoted $\cng_{G,T}(e)$,
as the number of edges between the two connected components $T_{1}$ and $T_{2}$ of $T-e$;
i.e., $\cng_{G,T}(e) = |E(V(T_{1}), V(T_{2}))|$.
We define the \emph{congestion} of $T$ as $\cng_{G}(T) = \max_{e \in E(T)} \cng_{G,T}(e)$.
The \emph{spanning tree congestion} of $G$, denoted $\stc(G)$, is the minimum congestion over all spanning trees of $G$.

Let $T$ be a tree. A vertex $v \in V(T)$ is a \emph{leaf} of $T$ if $\deg_{T}(v) \le 1$.
If $\deg_{T}(v) \ge 3$, then we call $v$ a \emph{branching vertex}.
A tree is a \emph{spider} if it has at most one branching vertex.
A tree is a \emph{star} if there is a vertex $c$, called the \emph{center}, adjacent to all other vertices;
i.e., a star is the complete bipartite graph $K_{1, s}$ for some $s \ge 0$.

An \emph{interval representation} of a graph $G = (V,E)$ is a set of closed intervals $\{I_{v} \mid v \in V\}$ on a line
such that $\{u,v\} \in E$ if and only if $I_{u} \cap I_{v} \ne \emptyset$.
An interval representation is \emph{proper} if no interval properly contains another.
A graph is an \emph{interval graph} if it has an interval representation,
and it is a \emph{proper interval graph} if it has a proper interval representation.
It is known that a graph is a proper interval graph if and only if it is a claw-free ($K_{1,3}$-free) interval graph~\cite{Roberts69}.

Although we use the concept of (linear) clique-width, we omit the definition since we do not need it in the proof.
See~\cite{FellowsRRS09} for its definition.

%%%%%%%%%%%%%%%%%%%%%%%%%%%%%%%%%%%%%%%%%%%%%%%%%%%%%%%%%%%%%%%%%%%%%%%%%%%%%%%%%%%%%%%%%%%%%%%%%%%%%%%%%%%%%%%%%%%%%%%%%%%%%%%%
%%%%%%%%%%%%%%%%%%%%%%%%%%%%%%%%%%%%%%%%%%%%%%%%%%%%%%%%%%%%%%%%%%%%%%%%%%%%%%%%%%%%%%%%%%%%%%%%%%%%%%%%%%%%%%%%%%%%%%%%%%%%%%%%
%%%%%%%%%%%%%%%%%%%%%%%%%%%%%%%%%%%%%%%%%%%%%%%%%%%%%%%%%%%%%%%%%%%%%%%%%%%%%%%%%%%%%%%%%%%%%%%%%%%%%%%%%%%%%%%%%%%%%%%%%%%%%%%%

\section{An auxiliary lemma}

In this section, we present an auxiliary technical lemma, which is a key tool in our proofs
and generalizes the fact that, for complete graphs, the minimum-congestion spanning trees are stars.
Intuitively, the lemma says that a spanning tree $T$ of low congestion cannot ``non-trivially split'' a highly connected vertex subset $S$ of large size,
and thus $T$ has a star-like (or spider-like) structure over $S$.
This implies that $T$ gives us an \emph{assignment} of the vertices not in $S$ to elements of $S$;
i.e., we assign $v \notin S$ to $s \in S$ if $v$ and $s$ belong to the same \emph{leg} in the star-like structure of $T$.
In \cref{sec:main-proof,sec:diam2}, we use this lemma in a reduction from \textsc{3-Partition},
with $\rho_{1} = 1$ and $\rho_{2} = 1/2$ in \cref{sec:main-proof} and with $\rho_{1} = (k+1)/(k+3)$ and $\rho_{2} = (k+3)/(2(k+1))$ in \cref{sec:diam2}.
We hope the lemma will be useful for other reductions or even for some algorithms on dense graphs.

Let $T$ be a spanning tree of $G$.
We say that an edge $e \in E(T)$ \emph{splits} $S \subseteq V$ into $S_{1}$ and $S_{2}$
if the two connected components $T_{1}$ and $T_{2}$ of $T-e$ satisfy $\{V(T_{1}) \cap S, V(T_{2}) \cap S\} = \{S_{1}, S_{2}\}$.
We say that $e$ \emph{non-trivially splits} $S$ if $e$ splits $S$ into $S_{1}$ and $S_{2}$ with $\min\{|S_{1}|, |S_{2}|\} \ge 2$.

\begin{lemma}
\label{lem:spider}
Let $G = (V,E)$ be a graph and $T$ a spanning tree of $G$ with congestion at most~$k$.
Let $\rho_{1}, \rho_{2} \in \mathbb{R}$ be positive numbers such that $\rho_{1} \rho_{2} \ge 1/2$ and $\rho_{2} \ge 1/2$.
Let $S \subseteq V$ be a set of at least four vertices such that $|S| \ge \rho_{1} k + 1$ and $\delta(G[S]) \ge \rho_{2} |S| +1$.
If $T'$ is the minimal subtree of $T$ containing all vertices of $S$,
then $T'$ is a spider with a branching vertex such that 
\begin{itemize}
  \setlength{\itemsep}{0pt}
  \item all degree-$2$ vertices of $T'$ (if any exists) belong to $V \setminus S$, and
  \item the degree of the branching vertex of $T'$ is at least $|S| - 1$.
\end{itemize}
\end{lemma}
\begin{proof}
We first show that no edge in $T$ non-trivially splits $S$.
Suppose to the contrary that $e \in E(T)$ non-trivially splits $S$ into $S_{1}$ and $S_{2}$.
Without loss of generality, we assume that $|S_{1}| \le |S_{2}|$, and thus $2 \le |S_{1}| \le |S|/2$.
Now we have
\begin{align*}
  \cng_{G,T}(e)
  &\ge
  |E(S_{1}, S_{2})|
  =
  \sum_{v \in S_{1}} (\deg_{G[S]}(v) - \deg_{G[S_{1}]}(v))
  \\
  &\ge |S_{1}| ((\rho_{2} |S| +1) - (|S_{1}|-1))
  =
  |S_{1}| ((\rho_{2} |S| +2) - |S_{1}|).
\end{align*}
Observe that the minimum value of the concave function $f(s_{1}) = s_{1} ((\rho_{2} |S| + 2) - s_{1})$ for the range $2 \le s_{1} \le |S|/2$ ($\le \rho_{2} |S|$)
is $f(2) = 2 \rho_{2} |S|$.
Thus we have
\[
  \cng_{G,T}(e) \ge |S_{1}| ((\rho_{2} |S|+2) - |S_{1}|) \ge 2\rho_{2} |S| \ge 2\rho_{2}  (\rho_{1} k + 1) = 2\rho_{2}\rho_{1} k + 2\rho_{2} \ge k + 1,
\]
which contradicts the assumption that $T$ has congestion at most $k$.

Now we can see that $T'$ is a spider as follows.
By its minimality, every leaf of $T'$ belongs to $S$.
Hence, if $T'$ contains two branching vertices,
then each edge in the path connecting these branching vertices non-trivially splits $S$.
Thus, $T'$ has at most one branching vertex.
If $T'$ is a path, then some edge in $T'$ non-trivially splits $S$ as $|S| \ge 4$.
Thus we can conclude that $T'$ contains exactly one branching vertex;
that is, $T'$ is a spider with at least three leaves, where all leaves belong to $S$.

Finally, we show that no degree-$2$ vertex of $T'$ belongs to $S$.
We consider the unique branching vertex of $T'$ as its root and denote it by $r$.
Observe that if a degree-$2$ vertex $v$ of $T'$ belongs to $S$,
then each edge in the $v$--$r$ path in $T'$ non-trivially splits $S$;
one set contains $v$ and its leaf descendant and the other set contains all other leaves.
This implies that each degree-$2$ vertex of $T'$ (if any exists) belongs to $V \setminus S$.
Now, since each vertex of $S$ is either a leaf or the root of $T'$, the degree of $r$ is at least $|S|-1$.
\end{proof}

%%%%%%%%%%%%%%%%%%%%%%%%%%%%%%%%%%%%%%%%%%%%%%%%%%%%%%%%%%%%%%%%%%%%%%%%%%%%%%%%%%%%%%%%%%%%%%%%%%%%%%%%%%%%%%%%%%%%%%%%%%%%%%%%
%%%%%%%%%%%%%%%%%%%%%%%%%%%%%%%%%%%%%%%%%%%%%%%%%%%%%%%%%%%%%%%%%%%%%%%%%%%%%%%%%%%%%%%%%%%%%%%%%%%%%%%%%%%%%%%%%%%%%%%%%%%%%%%%
%%%%%%%%%%%%%%%%%%%%%%%%%%%%%%%%%%%%%%%%%%%%%%%%%%%%%%%%%%%%%%%%%%%%%%%%%%%%%%%%%%%%%%%%%%%%%%%%%%%%%%%%%%%%%%%%%%%%%%%%%%%%%%%%

\section{Proper interval graphs}
\label{sec:main-proof}

In this section, we prove our main theorem, which is restated below.
%%%%%%%%%
\thmPint*
%%%%%%%%%

Since {\STC} clearly belongs to NP, we only need to prove its NP-hardness on the stated class of graphs.
To this end, we present a reduction from \textsc{3-Partition}.
The input of \textsc{3-Partition} consists of $3m$ positive integers $a_{1}, \dots, a_{3m}$ and another integer $B$
such that $\sum_{i \in [3m]} a_{i} = m B$ and $B/4 < a_{i} < B/2$ for each $i \in [3m]$.
It asks whether there exists a partition $(A_{1}, \dots, A_{m})$ of $[3m]$ such that $\sum_{j \in A_{i}} a_{j} = B$ for $i \in [m]$.
Note that the assumption $B/4 < a_{i} < B/2$ for $i \in [3m]$ implies that $|A_{i}| = 3$ whenever $\sum_{j \in A_{i}} a_{j} = B$.
It is known that \textsc{3-Partition} is strongly NP-complete; 
that is, it is NP-complete even if all integers $a_{i}$ are upper-bounded by some polynomial in $m$~\cite[SP15]{GareyJ79}.

\subsection{Construction}
Let $\mathcal{I} = \langle a_{1}, \dots, a_{3m}; B \rangle$ be an instance of \textsc{3-Partition} with $B/4 < a_{i} < B/2$ for each $i \in [3m]$. 
We may assume that $m \ge 4$, since instances with $m < 4$ can be solved in polynomial time.
By multiplying all numbers by $8m$ if necessary, we may also assume that $a_{i} \ge 8m$ for each $i \in [3m]$ (and thus $B = \frac{1}{m} \sum_{i \in [3m]} a_{i} \ge 24m$).
We assume without loss of generality that $a_{1} \le a_{2} \le \dots \le a_{3m}$.
For $i \in [a_{3m}]$, let $\Gamma_{i} = \{j \in [3m] \mid a_{j} \ge i\}$ and $\gamma_{i} = |\Gamma_{i}|$.
Since $a_{1} \le \dots \le a_{3m}$, it follows that $\Gamma_{i} = [3m - \gamma_{i} + 1, 3m]$.
By definition, we have $\gamma_{i} \ge \gamma_{i'}$ for $i < i'$.
Note that for $i \in [m]$, $\gamma_{i} = 3m$ holds since $a_{j} \ge 8m \ge m$ for all $j \in [3m]$.

To construct an instance $\langle G, k \rangle$ of {\STCfull},
we set $k = 3B$ and construct $G = (V, E)$ with $V = X \cup Y \cup Z$ as follows.
\begin{itemize}
  \setlength{\itemsep}{0pt}
  \item 
    $X = \{x_{i} \mid i \in [3m]\}$,
    $Y = \{y_{i} \mid i \in [a_{3m}]\}$, and
    $Z = \{z_{i} \mid i \in [k-a_{3m}+1]\}$.
  \item $X$, $Y$, and $Z$ are cliques.
  \item For $i \in [3m]$, $x_{i}$ is adjacent to $\{y_{j} \mid j \in [a_{i}]\}$.
  \item For $i \in [m]$, $y_{i}$ is adjacent to $\{z_{j} \mid j \in [|Z| - B  - 12m + 15]\}$.
  \item For $i \in [m+1, a_{3m}]$, $y_{i}$ is adjacent to $\{z_{j} \mid j \in [|Z| - \gamma_{i}]\}$.
\end{itemize}
The two expressions used as upper bounds for subsets of $Z$ are positive. 
Indeed, since $a_{3m} < B/2$, $B \ge 24m$, and $\gamma_{i} \le 3m$, we have $|Z|-B-12m+15 > 0$ and $|Z|-\gamma_{i} > 0$.
Note that, for $i \in [a_{3m}]$, the neighborhood of $y_{i}$ in $X$ is $\{x_{j} \mid j \in [3m - \gamma_{i} + 1, 3m]\}$,
and thus $y_{i}$ has $\gamma_{i}$ neighbors in $X$.
See \cref{fig:reduction}.

\begin{figure}[tb]
  \centering
  \includegraphics{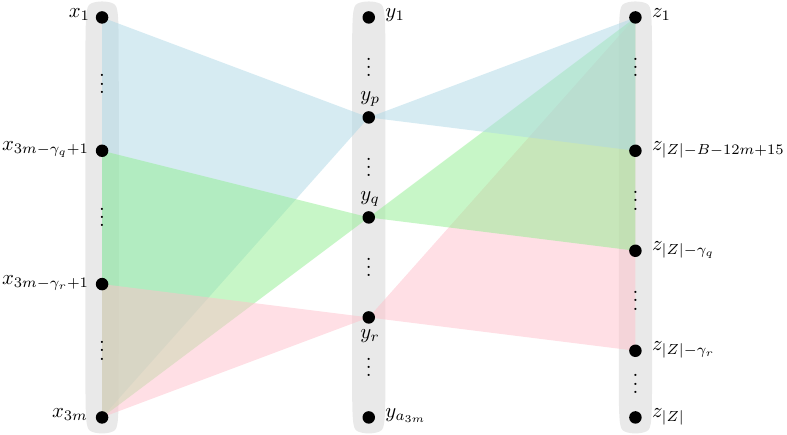}
  \caption{The neighborhoods of $y_{p}$, $y_{q}$, and $y_{r}$ in $X$ and $Z$, where $p \in [m]$ and $m < q < r \le a_{3m}$.}
  \label{fig:reduction}
\end{figure}

\subsection{Graph class}

We now show that $G$ belongs to the target graph class.
The next observation is useful for showing that $G$ is a proper interval graph (see \cref{fig:reduction}).

\begin{observation}
\label{obs:chain-property}
If $i < i'$, then $N_{G}(y_{i}) \cap Z \subseteq N_{G}(y_{i'}) \cap Z$.
\end{observation}
\begin{proof}
If $i, i' \in [m]$, then $y_{i}$ and $y_{i'}$ have the same neighbors in $Z$.
If $i, i' \in [m+1, a_{3m}]$, then we have $\gamma_{i} \ge \gamma_{i'}$, and thus $N_{G}(y_{i}) \cap Z \subseteq N_{G}(y_{i'}) \cap Z$.
Finally, assume that $i \in [m]$ and $i' \in [m+1, a_{3m}]$.
By $\gamma_{i'} \le 3m$ and $B \ge 24m$, it follows that $|Z| - B - 12m+ 15 \le |Z| - \gamma_{i'}$,
and thus $N_{G}(y_{i}) \cap Z \subseteq N_{G}(y_{i'}) \cap Z$ holds in this case as well.
\end{proof}

\begin{lemma}
\label{lem:G-pint}
$G$ is a proper interval graph of linear clique-width at most~$4$ and diameter~$3$.
\end{lemma}
\begin{proof}
A linear ordering $\prec$ of vertices is a \emph{proper interval ordering} if for every triple $u,v, w$ of vertices,
$u \prec v \prec w$ and $\{u,w\} \in E$ imply $\{u,v\}, \{v,w\} \in E$.
It is known that a graph is a proper interval graph if and only if it admits a proper interval ordering~\cite{Looges93}.
We show that $(x_{1}, \dots, x_{|X|},\; y_{1}, \dots, y_{|Y|},\; z_{1}, \dots, z_{|Z|})$ is a proper interval ordering of $G$.

Let $u, v, w$ be vertices with $u \prec v \prec w$ and $\{u,w\} \in E$.
Since there is no edge between $X$ and $Z$, we have either $\{u,w\} \subseteq X \cup Y$ or $\{u,w\} \subseteq Y \cup Z$.
If $u$ and $w$ belong to the same set $W \in \{X,Y,Z\}$, then $v \in W$ as well, and thus $u,v,w$ form a clique.
The remaining are the case of $u \in X$, $w \in Y$ and of $u \in Y$, $w \in Z$.

First consider the case $u = x_{i} \in X$ and $w = y_{j} \in Y$.
Since $\{x_{i}, y_{j}\} \in E$, $x_{i}$ is adjacent to all vertices $y_{j'}$ with $j' < j$.
Also, every vertex $x_{i'}$ with $i < i'$ is adjacent to $y_{j}$ since $a_{i} \le a_{i'}$.
Since $v \in \{x_{i'} \mid i' > i\} \cup \{y_{j'} \mid j' < j\}$, we have $\{u,v\}, \{v,w\} \in E$.
(Recall that $X$ and $Y$ are cliques.)

The proof for the remaining case, where $u = y_{i} \in Y$ and $w = z_{j} \in Z$, is almost identical to the previous one.
Since $\{y_{i}, z_{j}\} \in E$, $y_{i}$ is adjacent to all vertices $z_{j'}$ with $j' < j$.
Also, every vertex $y_{i'}$ with $i < i'$ is adjacent to $z_{j}$ since $N_{G}(y_{i}) \cap Z \subseteq N_{G}(y_{i'}) \cap Z$ by \cref{obs:chain-property}.
Since $v \in \{y_{i'} \mid i' > i\} \cup \{z_{j'} \mid j' < j\}$, we have $\{u,v\}, \{v,w\} \in E$.
(Recall that $Y$ and $Z$ are cliques.)

We next show that $G$ has linear clique-width at most~$4$.
Heggernes, Meister, and Papadopoulos~\cite[Theorem~3.1]{HeggernesMP09} showed that the linear clique-width of
a proper interval graph is at most the independence number plus~$1$.
(They stated their result only for clique-width, but their proof actually constructs a linear expression, 
and thus gives the same upper bound for linear clique-width.)
Since the vertex set $V$ of $G$ consists of three cliques $X$, $Y$, $Z$, 
the independence number of $G$ is at most~$3$, and thus, its linear clique-width is at most~$4$.

Finally, we show that $G$ has diameter~$3$.
Recall that $X$, $Y$, $Z$ are cliques,
that $y_{1}$ is adjacent to every vertex of $X$, and that $z_{1}$ is adjacent to every vertex of $Y$.
If two vertices $u,v$ both belong to the same set $X$, $Y$, or $Z$, then they are adjacent.
For $x \in X$ and $y \in Y$, there is a path $(x,y_{1},y)$ (possibly with $y = y_{1}$) of length at most~$2$.
Similarly, for $y \in Y$ and $z \in Z$, there is a path $(y,z_{1},z)$ (possibly with $z = z_{1}$) of length at most~$2$.
For $x \in X$ and $z \in Z$, there is a path $(x,y_{1},z_{1},z)$ (possibly with $z = z_{1}$) of length at most~$3$.
Thus, $G$ has diameter at most~$3$.
Finally, $z_{|Z|}$ has neighbors only in $Z$, and thus it has distance exactly~$3$ to vertices in $X$.
\end{proof}

\subsection{Equivalence}

To complete the proof of \cref{thm:stc-pint},
we show that $\stc(G) \le k$ if and only if $\mathcal{I}$ is a yes-instance of \textsc{3-Partition}.
We prove one direction in \cref{lem:3partition->stc} and the other direction in \cref{lem:stc->3partition}.

\paragraph{Vertex degrees}
From the construction of $G$, we can easily compute the degree of each vertex,
which will be helpful when we compute the congestion of an edge in a spanning tree.
We summarize them below and use them in later discussions without explicit references.
\begin{itemize}
  \setlength{\itemsep}{0pt}
  \item For $i \in [3m]$, $\deg_{G}(x_{i}) = (|X|-1) + a_{i} = a_{i} + 3m - 1$ ($< k$).
  \item For $i \in [m]$, $\deg_{G}(y_{i}) = |X| + (|Y|-1) + (|Z| - B  - 12m + 15) = k - B - 9m + 15$.
  \item For $i \in [m+1, a_{3m}]$, $\deg_{G}(y_{i}) = \gamma_{i} + (|Y|-1) + (|Z| - \gamma_{i}) = k$.
  \item For $z_{i} \in Z$, $\deg_{G}(z_{i}) \in [|Z|-1, |Y \cup Z| -1] = [k - a_{3m}, k]$.
\end{itemize}
We can see that $\Delta(G) = k$. 
Note that $\delta(G[Y \cup Z]) = k - B - 12m + 15\ge (k+1)/2 + 1$ as $k = 3B$ and $B \ge 24m$.

\begin{lemma}
\label{lem:3partition->stc}
If $\mathcal{I}$ is a yes-instance of \textsc{3-Partition}, then $\stc(G) \le k$.
\end{lemma}
\begin{proof}
Let $(A_{1},\dots,A_{m})$ be a partition of $[3m]$
such that $\sum_{j \in A_{i}} a_{j} = B$ and $|A_{i}| = 3$ for each $i \in [m]$.
We construct a spanning tree $T$ of $G$ as follows:
\[
  E(T) = \{\{z_{1}, w\} \mid w \in Y \cup (Z \setminus \{z_{1}\})\} \cup \bigcup_{i \in [m]} \{\{y_{i}, x_{j}\} \mid j \in A_{i}\}.
\]
In $T$, the edges in $\{\{z_{1}, y_{i}\} \mid i \in [m]\}$ are the inner edges and all other edges are leaf edges.
Since a leaf edge has congestion at most $\Delta(G) = k$, it suffices to show that $\cng_{G,T}(\{z_{1}, y_{i}\}) \le k$ for each $i \in [m]$.
Fix $i \in [m]$. Since $\{y_{i}\} \cup \{x_{j} \mid j \in A_{i}\}$ is a $4$-clique,
we have $\cng_{G,T}(\{z_{1}, y_{i}\}) = k$ as follows:
\begin{align*}
  \cng_{G,T}(\{z_{1}, y_{i}\})
  &= 
  \deg_{G}(y_{i}) + \sum_{j \in A_{i}} \deg_{G}(x_{j}) - 2 \cdot \binom{4}{2}
  \\
  &= 
  (k - B - 9m + 15) + \sum_{j \in A_{i}} (a_{j} + 3m - 1) -12
  \\
  &= 
  (k - B - 9m + 15) + (B + 9m-3) -12
  \\
  &= 
  k.
  \qedhere
\end{align*}
\end{proof}

\begin{lemma}
\label{lem:stc->3partition}
If $\stc(G) \le k$, then $\mathcal{I}$ is a yes-instance of \textsc{3-Partition}.
\end{lemma}
\begin{proof}
Let $T$ be a spanning tree of $G$ with congestion at most $k$
and let $T'$ be the inclusion-wise minimal subtree of $T$ that contains all vertices of $Y \cup Z$.
Let $\rho_{1} = 1$, $\rho_{2} = 1/2$, and $S = Y \cup Z$.
We can see that $|S| = |Y \cup Z| = k+1 = \rho_{1} k + 1$ ($\ge 4$)
and $\delta(G[S]) = \delta(G[Y \cup Z]) \ge (k+1)/2 + 1 = \rho_{2} |Y \cup Z| + 1 = \rho_{2} |S| + 1$.
Hence, by \cref{lem:spider}, $T'$ is a spider such that all degree-$2$ vertices belong to $X$ ($= V \setminus (Y \cup Z)$)
and the branching vertex has degree at least $|Y \cup Z| - 1 = k$.

\begin{claim}
\label{clm:star}
$T'$ is a star centered at a vertex of $Z$.
\end{claim}
\begin{subproof}[Proof of \cref{clm:star}]
Let $r$ be the branching vertex of $T'$.
We first show that $r \in Z$.
As $\deg_{T'}(r) \ge k$, we have $r \in Y \cup Z$.
Observe that no vertex $y_{i} \in Y$ is adjacent to $z_{|Z|}$ since $|Z| - B - 12m + 15 < |Z|$ and $|Z| - \gamma_{i} < |Z|$ for all $i$.
This implies that $N_{G}[z_{|Z|}] = Z$.
Thus we have $\deg_{G}(z_{|Z|}) = |Z| - 1 = k-a_{3m} < k$, and so $z_{|Z|} \ne r$.
Moreover, since $N_{G}[z_{|Z|}] = Z$, the $r$--$z_{|Z|}$ path in $T'$ contains a vertex of $Z \setminus \{z_{|Z|}\}$.
This implies that $r \in Z$ since every degree-$2$ vertex of $T'$ belongs to $X$.

It remains to show that $V(T') = Y \cup Z$. 
This will imply that $T'$ is a star, because no vertex of $Y \cup Z$ has degree~$2$ in $T'$.
Suppose to the contrary that $T'$ contains a vertex $x \in X$.
Since $r \in Z$ and there is no edge between $X$ and $Z$, the $r$--$x$ path in $T'$ contains a vertex of $Y$, say $y$.
This vertex $y$ has degree~$2$ in $T'$, contradicting the assumption that all degree-$2$ vertices of $T'$ belong to $X$.
\end{subproof}

Let $z^{*} \in Z$ be the center of the star $T'= T[Y \cup Z]$.
We consider $z^{*}$ as the root of $T$.

\begin{claim}
\label{clm:z-leaves}
Each $z \in Z \setminus \{z^{*}\}$ is a leaf of $T$.
\end{claim}
\begin{subproof}[Proof of \cref{clm:z-leaves}]
Observe that $N_{G}(z) \subseteq Y \cup Z \subseteq N_{T}[z^{*}]$.
Hence, if $z$ has a neighbor $w \ne z^{*}$ in $T$,
then the three vertices $z^{*}$, $z$, $w$ form a triangle in $T$, a contradiction.
\end{subproof}

Observe that for each $y_{i} \in Y$, its descendants in $T$ belong to $X$.
Let $X_{i} \subseteq X$ be the set of descendants of $y_{i}$ in $T$.

\begin{claim}
\label{clm:y-leaves}
For $i \in [m+1, a_{3m}]$, $y_{i}$ is a leaf of $T$.
\end{claim}
\begin{subproof}[Proof of \cref{clm:y-leaves}]
Suppose to the contrary that $y_{i}$ is not a leaf of $T$ (i.e., $X_{i} \ne \emptyset$) for some $i \in [m+1, a_{3m}]$.
This leads to the contradiction that $\cng_{G,T}(\{z^{*}, y_{i}\}) > k$, as shown below:
\begin{align*}
  \cng_{G,T}(\{z^{*}, y_{i}\})
  &\ge
  \deg_{G}(y_{i}) + \sum_{x_{j} \in X_{i}} \deg_{G}(x_{j}) - 2 \cdot \binom{|X_{i}|+1}{2}
  \\
  &=
  k + \sum_{x_{j} \in X_{i}} ((a_{j} + 3m-1) - (|X_{i}|+1))
  \\
  &>
  k,
\end{align*}
where the last inequality holds
because $X_{i} \ne \emptyset$ and 
$(a_{j} + 3m-1) - (|X_{i}|+1) = (a_{j} -2) + (3m - |X_{i}|) > 0$ for every $j \in [3m]$.
\end{subproof}

\cref{clm:z-leaves,clm:y-leaves} imply that $X_{1}, \dots, X_{m}$ is a partition of~$X$.
Now we show that this partition gives a solution to \textsc{3-Partition}.

\begin{claim}
\label{clm:correct-sum}
For $i \in [m]$, $\sum_{x_{j} \in X_{i}} a_{j} = B$.
\end{claim}
\begin{subproof}[Proof of \cref{clm:correct-sum}]
Since $\sum_{x_{j} \in X} a_{j} = mB$, it suffices to show that $\sum_{x_{j} \in X_{i}} a_{j} \le B$ for every $i \in [m]$.
Suppose to the contrary that $\sum_{x_{j} \in X_{i}} a_{j} > B$ for some $i \in [m]$.
The congestion of the edge $\{z^{*}, y_{i}\}$ can be lower-bounded as
\begin{align*}
  \cng_{G,T}(\{z^{*}, y_{i}\})
  &\ge
  \deg_{G}(y_{i}) + \sum_{x_{j} \in X_{i}} \deg_{G}(x_{j}) - 2 \cdot \binom{|X_{i}|+1}{2}
  \\
  &=
  (k-B-9m+15) + \sum_{x_{j} \in X_{i}} (a_{j}+3m-1) - |X_{i}| (|X_{i}|+1)
  \\
  &=
  (k-B-9m+15) + \sum_{x_{j} \in X_{i}} a_{j} + |X_{i}| ((3m -2) - |X_{i}|).
\end{align*}

Now we obtain a lower bound on $\sum_{x_{j} \in X_{i}} a_{j} + |X_{i}| ((3m -2) - |X_{i}|)$.
First observe that $|X_{i}| \ge 3$ since $\sum_{x_{j} \in X_{i}} a_{j} > B$ and $a_{j} < B/2$ for $j \in [3m]$.
Hence $3 \le |X_{i}| \le 3m$ holds, and it follows that $|X_{i}| ((3m -2) - |X_{i}|) \ge 3m ((3m -2) - 3m) = -6m$.
If $\sum_{x_{j} \in X_{i}} a_{j} \le 2B$, then $|X_{i}| \le 7$ holds as $a_{j} > B/4$ for $j \in [3m]$.
Hence, in this case, we have a smaller range of $3 \le |X_{i}| \le 7$ ($\le 3m-5$), which implies $|X_{i}| ((3m -2) - |X_{i}|) \ge 3 ((3m -2) - 3) = 9m -15$.
Now we obtain
\begin{align*}
  \sum_{x_{j} \in X_{i}} a_{j} + |X_{i}| ((3m -2) - |X_{i}|)
  &>
  \begin{cases}
    B + 9m-15 & (\sum_{x_{j} \in X_{i}} a_{j} \le 2B) \\
    2B -6m & (\sum_{x_{j} \in X_{i}} a_{j} > 2B) \\
  \end{cases}
  \\
  &\ge
  B + 9m-15,
\end{align*}
where the last inequality holds by the assumption $B \ge 24m$.

Combining the discussions above, we reach the contradiction that $\cng_{G,T}(\{z^{*}, y_{i}\}) > k$ as follows:
\begin{align*}
  \cng_{G,T}(\{z^{*}, y_{i}\})
  &\ge
  (k-B-9m+15) + \sum_{x_{j} \in X_{i}} a_{j} + |X_{i}| ((3m -2) - |X_{i}|)
  \\
  &>
  (k-B-9m+15) + B + 9m-15
  \\
  &= k.
  \qedhere
\end{align*}
\end{subproof}

\cref{clm:correct-sum} implies that the partition $A_{1}, \dots, A_{m}$ of $[3m]$ with $A_{i} = \{j \in [3m] \mid x_{j} \in X_{i}\}$ gives a solution to \textsc{3-Partition}.
This completes the proof.
\end{proof}

%%%%%%%%%%%%%%%%%%%%%%%%%%%%%%%%%%%%%%%%%%%%%%%%%%%%%%%%%%%%%%%%%%%%%%%%%%%%%%%%%%%%%%%%%%%%%%%%%%%%%%%%%%%%%%%%%%%%%%%%%%%%%%%%
%%%%%%%%%%%%%%%%%%%%%%%%%%%%%%%%%%%%%%%%%%%%%%%%%%%%%%%%%%%%%%%%%%%%%%%%%%%%%%%%%%%%%%%%%%%%%%%%%%%%%%%%%%%%%%%%%%%%%%%%%%%%%%%%
%%%%%%%%%%%%%%%%%%%%%%%%%%%%%%%%%%%%%%%%%%%%%%%%%%%%%%%%%%%%%%%%%%%%%%%%%%%%%%%%%%%%%%%%%%%%%%%%%%%%%%%%%%%%%%%%%%%%%%%%%%%%%%%%

\section{Graphs of diameter 2}
\label{sec:diam2}

In this section, we make small changes to the construction in \cref{sec:main-proof} and prove \cref{thm:diam2}.

\thmDiamTwo*

Let $\mathcal{I} = \langle a_{1}, \dots, a_{3m}; B \rangle$ be an instance of \textsc{3-Partition} with $m \ge 4$,
$B/4 < a_{1} \le \dots \le a_{3m} < B/2$, and $a_{i} \ge 8m$ for every $i \in [3m]$.
Let $\langle G = (V,E), k \rangle$ be the instance of \STC{} constructed from $\mathcal{I}$ in \cref{sec:main-proof}.
Recall that $k = 3B$ and $V = X \cup Y \cup Z$, where
$X = \{x_{i} \mid i \in [3m]\}$,
$Y = \{y_{i} \mid i \in [a_{3m}]\}$, and
$Z = \{z_{i} \mid i \in [k-a_{3m}+1]\}$.

From $\langle G, k \rangle$, we construct a new instance $\langle H, k' \rangle$ of \STC{} as follows.
Set $k' = k + 3$.
To construct $H$ from $G$,
add a vertex $p$ and a set $D$ of $a_{3m} -3m +2$ vertices,
and then
make $p$ adjacent to all vertices in $D \cup X \cup Z$
and make $z_{1}$ adjacent to all vertices in $D$.
See \cref{fig:reduction2}.

\begin{figure}[tb]
  \centering
  \includegraphics{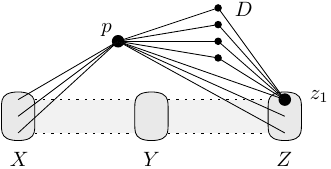}
  \caption{The graph $H$ constructed from $G$ with additional vertices $\{p\} \cup D$.}
  \label{fig:reduction2}
\end{figure}

The degrees of vertices in $H$ compare to those in $G$ as follows:
$\deg_{H}(x) = \deg_{G}(x) + 1$ for $x \in X$;
$\deg_{H}(y) = \deg_{G}(y)$ for $y \in Y$;
$\deg_{H}(z_{1}) = \deg_{G}(z_{1}) + |D| + 1$;
$\deg_{H}(z) = \deg_{G}(z) + 1$ for $z \in Z \setminus \{z_{1}\}$.
For new vertices, it follows that 
$\deg_{H}(d) = 2$ for $d \in D$
and
\[
\deg_{H}(p) = |X| + |Z| + |D| = 3m + (k-a_{3m}+1) + (a_{3m}-3m+2) = k+3 = k'. 
\]

\begin{lemma}
\label{lem:diam2}
$H$ has diameter~$2$.
\end{lemma}
\begin{proof}
In the original graph $G$, two vertices $u,v$ satisfy $\dist_{G}(u,v) > 2$ only if $u \in X$ and $v \in Z$, or $u \in Z$ and $v \in X$.
In $H$, such a pair of vertices has $p$ as a common neighbor. This implies that $\dist_{H}(u,v) \le 2$ for $u, v \in X \cup Y \cup Z$.
The vertex $p$ is at distance at most~$2$ from every vertex of $Y$ through $z_{1}$ as $z_{1}$ is adjacent to every vertex of $Y$.
Each vertex $d \in D$ is adjacent to $p$, is at distance at most~$2$ from every vertex of $X$ through $p$, and is at distance at most~$2$ from every vertex of $Y \cup Z$ through $z_{1}$. 
Any two vertices of $D$ have the common neighbor $p$. 
Thus every two vertices of $H$ are at distance at most~$2$. 
Finally, $p$ is not adjacent to any vertex of $Y$, so the diameter is exactly~$2$.
\end{proof}

In the following, we show that $\mathcal{I}$ is a yes-instance of \textsc{3-Partition} if and only if $\langle H, k' \rangle$ is a yes-instance of \STC.

\begin{lemma}
If $\mathcal{I}$ is a yes-instance of \textsc{3-Partition}, then $\langle H, k' \rangle$ is a yes-instance of \STC.
\end{lemma}
\begin{proof}
Let $(A_{1}, \dots, A_{m})$ be a partition of $[3m]$ with
$|A_{i}| = 3$ and $\sum_{j \in A_{i}} a_{j} = B$ for all $i \in [m]$.
We construct a spanning tree $T'$ of $H$
by first taking the spanning tree $T$ of $G$ used in the proof of \cref{lem:3partition->stc} 
and then adding the leaf edges $\{z_{1}, p\}$ and $\{z_{1}, d\}$ for all $d \in D$.

Since $z_{1}$ is the unique vertex with degree more than $k'$ in $H$ and $z_{1}$ is not a leaf of $T'$, 
all leaf edges of $T'$ have congestion at most $k'$.

Observe that each non-leaf edge of $T'$ has the form $\{z_{1}, y_{i}\}$ with $i \in [m]$.
For such an edge, the component not containing $z_{1}$ is $\{y_{i}\} \cup \{x_{j} \mid j \in A_{i}\}$.
It follows that 
\[
  \cng_{H,T'}(\{z_{1},y_{i}\}) = \cng_{G,T}(\{z_{1}, y_{i}\}) + 3 = k+3 = k',
\]
where the increase $+3$ comes from the edges between $p$ and the three vertices in $\{x_{j} \mid j \in A_{i}\}$.
\end{proof}

\begin{lemma}
If $\langle H, k' \rangle$ is a yes-instance of \STC, then $\mathcal{I}$ is a yes-instance of \textsc{3-Partition}.
\end{lemma}
\begin{proof}
Let $T$ be a spanning tree of $H$ with congestion at most~$k'$
and let $T'$ be the inclusion-wise minimal subtree of $T$ that contains all vertices of $\{p\} \cup Y \cup Z$.
Let
$\rho_{1} = (k+1)/(k+3)$,
$\rho_{2} = (k+3)/(2(k+1))$, and
$S = \{p\} \cup Y \cup Z$.
It follows that $|S| = k+2 = \rho_{1} k' + 1 \ge 4$, $\rho_{1} \rho_{2} = 1/2$, and $\rho_{2} \ge 1/2$.
We can see that $\deg_{H[S]}(p) = \min_{z \in Z} \deg_{H[S]}(z) = |Z| = k-a_{3m}+1$
and $\min_{y \in Y}\deg_{H[S]}(y) = k-B-12m+15$.
It follows that $\delta(H[S]) = k-B-12m+15 \ge k/2 + 15$.
On the other hand, we have $\rho_{2} |S| + 1 = (k+2)(k+3)/(2(k+1)) + 1 = k/2 + 1/(k+1) + 3$.
Hence, $\delta(H[S]) \ge \rho_{2} |S| + 1$ holds.
Now, \cref{lem:spider} implies that $T'$ is a spider whose branching vertex has degree at least $|S|-1=k+1$.

\begin{claim}
\label{clm:spider-diam2}
The branching vertex of $T'$ belongs to $Z$.
\end{claim}
\begin{subproof}[Proof of \cref{clm:spider-diam2}]
Let $r$ be the branching vertex of $T'$.
Since $\deg_{T'}(r) \ge k+1$, we have $r \in \{p\} \cup Z$.
Suppose to the contrary that $r = p$.
For $y \in Y$, let $P_{y}$ be the $p$--$y$ path in $T'$.
Since $\{p,y\} \notin E(H)$ and $P_{y}$ cannot have vertices of $S$ as internal vertices,
$P_{y}$ has to contain a vertex of $X$.
Since $T'$ is a spider, each $P_{y}$ contains a distinct member of $X$.
This contradicts that $|Y| = a_{3m} \ge 8m > 3m = |X|$.
\end{subproof}

In the following, let $z^{*} \in Z$ be the branching vertex of $T'$
and consider $z^{*}$ as the root of $T$.

\begin{claim}
\label{clm:almost-star-diam2}
Every vertex of $Y \cup (Z \setminus \{z^{*}\})$ is adjacent to $z^{*}$ in $T$.
\end{claim}
\begin{subproof}[Proof of \cref{clm:almost-star-diam2}]
Let $s \in Y \cup (Z \setminus \{z^{*}\})$.
If the $z^{*}$--$s$ path $P$ in $T'$ has length at least~$2$, then the neighbor of $z^{*}$ in $P$ has to lie outside $S$.
Only the vertices in $D$ can be such a neighbor, but then the next vertex in $P$ has to be $p$, a contradiction.
\end{subproof}

\begin{claim}
\label{clm:no-p-dec-diam2}
In $T$, no vertex of $X$ is a descendant of $p$.
\end{claim}
\begin{subproof}[Proof of \cref{clm:no-p-dec-diam2}]
Let $A$ be the vertex set of the subtree of $T$ rooted at $p$, and let $e$ be the edge of $T$ connecting $p$ and its parent.
\cref{clm:almost-star-diam2} implies that $A \subseteq \{p\} \cup X \cup D$.
We can show that $\cng_{H,T}(e) \ge k' + |A \cap X|$ as follows:
\begin{align*}
  \cng_{H,T}(e)
  &=
  \sum_{v \in A} \deg_{H}(v) - 2|E(H[A])|
  \\
  &=
  k'
  + 2 |A \cap D|
  + \sum_{x_{i} \in A \cap X} (a_{i} + 3m)
  - 2(|A \cap D| + |A \cap X| + |A \cap X| (|A \cap X|-1)/2)
  \\
  &\ge
  k'
  + \sum_{x_{i} \in A \cap X} (11m)
  - |A \cap X|(|A \cap X|+1)
  =
  k'
  + |A \cap X| (11m - |A \cap X|-1)
  \\
  &\ge
  k' + |A \cap X|,
\end{align*}
where the inequalities use $a_{i} \ge 8m$ and $|A \cap X| \le 3m$.
Since $\cng_{H,T}(e) \le k'$, it follows that $|A \cap X| = 0$.
\end{subproof}

\begin{claim}
\label{clm:pD-diam2}
No vertex of $\{p\} \cup D$ is a descendant of a vertex of $Y$ in $T$.
\end{claim}
\begin{subproof}[Proof of \cref{clm:pD-diam2}]
We first show that no vertex of $Y$ lies on the $z^{*}$--$p$ path in $T$.
Since $T'$ is a subtree of $T$ that contains both $z^{*}$ and $p$, this path lies in $T'$,
and thus all its internal vertices have degree~$2$ in $T'$
and belong to $V(H) \setminus S = X \cup D$ by \cref{lem:spider}.
Assume that the path has at least one internal vertex.
Since there is no edge between $X$ and $Z$, the internal vertex adjacent to $z^{*}$ belongs to $D$,
and thus $z^{*} = z_{1}$ as $z_{1}$ is the only vertex of $Z$ with neighbors in $D$.
Now, since $N_{H}(d) = \{p, z_{1}\}$ for each $d \in D$, the vertex following this internal vertex is $p$.
Hence, the $z^{*}$--$p$ path in $T$ is either the single edge $\{z^{*}, p\}$
or a path $(z_{1}, d, p)$ with $d \in D$.
In particular, no vertex of $Y$ lies on this path,
and thus $p$ is not a descendant of a vertex of $Y$.

Next, let $d \in D$.
Since $N_{H}(d) = \{p, z_{1}\}$, the parent of $d$ in $T$ is $p$ or $z_{1}$.
If the parent is $p$, then the statement for $d$ follows from the one for $p$ shown above.
If the parent is $z_{1}$, then $z_{1}$ is $z^{*}$ itself or a child of $z^{*}$ by \cref{clm:almost-star-diam2},
and thus no vertex of $Y$ is an ancestor of $d$.
\end{subproof}

Having handled the vertices in $\{p\} \cup D$, the rest of the proof is almost the same as that of \cref{lem:stc->3partition}.
For completeness, we present the remaining steps in full.

By~\cref{clm:no-p-dec-diam2}, the vertices of $X$ are partitioned among the subtrees rooted at vertices of $Y$.
For $i \in [a_{3m}]$, let $X_{i} \subseteq X$ be the set of descendants of $y_{i}$ in $T$.
By \cref{clm:almost-star-diam2,clm:pD-diam2}, every descendant of $y_{i}$ belongs to $X$.
In particular, if $y_{i}$ is not a leaf of $T$,
then the component of $T - \{z^{*}, y_{i}\}$ not containing $z^{*}$ is exactly $\{y_{i}\} \cup X_{i}$.

\begin{claim}
\label{clm:y-leaves-diam2}
For $i \in [m+1, a_{3m}]$, $y_{i}$ is a leaf of $T$.
\end{claim}
\begin{subproof}[Proof of \cref{clm:y-leaves-diam2}]
Suppose that $y_{i}$ is not a leaf of $T$ (i.e., $X_{i} \ne \emptyset$) for some $i \in [m+1, a_{3m}]$.
This gives the contradiction that $\cng_{H,T}(\{z^{*}, y_{i}\}) > k'$ ($= k+3$) as follows:
\begin{align*}
  \cng_{H,T}(\{z^{*}, y_{i}\})
  &\ge
  \deg_{H}(y_{i}) + \sum_{x_{j} \in X_{i}} \deg_{H}(x_{j}) - 2 \cdot \binom{|X_{i}|+1}{2}
  \\
  &=
  k + \sum_{x_{j} \in X_{i}} ((a_{j} + 3m) - (|X_{i}|+1))
  \\
  &>
  k+3,
\end{align*}
where the last inequality holds
because $X_{i} \ne \emptyset$ and 
$(a_{j} + 3m) - (|X_{i}|+1) \ge 11m - (3m+1) > 3$ for every $j \in [3m]$.
\end{subproof}

\cref{clm:y-leaves-diam2} implies that $X_{1}, \dots, X_{m}$ is a partition of $X$.
It remains to show that this partition gives a solution to \textsc{3-Partition}.
As in the proof of \cref{lem:stc->3partition}, it suffices to show that $\sum_{x_{j} \in X_{i}} a_{j} \le B$ for every $i \in [m]$.
Suppose to the contrary that $\sum_{x_{j} \in X_{i}} a_{j} > B$ for some $i \in [m]$.
It follows that $|X_{i}| \ge 3$ as $a_{j} < B/2$ for every $j \in [3m]$.
Now it follows that $\cng_{H,T}(\{z^{*},y_{i}\}) > k'$ ($= k+3$) as follows:
\begin{align*}
  \cng_{H,T}(\{z^{*},y_{i}\})
  &\ge
  \deg_{H}(y_{i}) + \sum_{x_{j} \in X_{i}} \deg_{H}(x_{j}) -2 \cdot \binom{|X_{i}|+1}{2}
  \\
  &=
  (k-B-9m+15) + \sum_{x_{j} \in X_{i}} (a_{j} + 3m) - |X_{i}| (|X_{i}|+1)
  \\
  &=
  (k-B-9m+15) + \sum_{x_{j} \in X_{i}} a_{j} + |X_{i}| ((3m-2)-|X_{i}|) + |X_{i}|
  \\
  &>
  k + |X_{i}|
  \\
  &\ge
  k + 3,
\end{align*}
where the second-to-last inequality follows from the argument in the proof of \cref{clm:correct-sum}.
\end{proof}

%%%%%%%%%%%%%%%%%%%%%%%%%%%%%%%%%%%%%%%%%%%%%%%%%%%%%%%%%%%%%%%%%%%%%%%%%%%%%%%%%%%%%%%%%%%%%%%%%%%%%%%%%%%%%%%%%%%%%%%%%%%%%%%%
%%%%%%%%%%%%%%%%%%%%%%%%%%%%%%%%%%%%%%%%%%%%%%%%%%%%%%%%%%%%%%%%%%%%%%%%%%%%%%%%%%%%%%%%%%%%%%%%%%%%%%%%%%%%%%%%%%%%%%%%%%%%%%%%
%%%%%%%%%%%%%%%%%%%%%%%%%%%%%%%%%%%%%%%%%%%%%%%%%%%%%%%%%%%%%%%%%%%%%%%%%%%%%%%%%%%%%%%%%%%%%%%%%%%%%%%%%%%%%%%%%%%%%%%%%%%%%%%%

\section{Proper interval graphs of diameter 2}
\label{sec:cochain}
In this section, we show that proper interval graphs of diameter at most~$2$ are cochain graphs.
A bipartite graph is a \emph{chain graph} if it does not contain $2K_{2}$ (the graph formed by two independent edges) as an induced subgraph.
A graph is a \emph{cochain graph} if its complement is a chain graph.
The $\{3K_{1}, C_{4}, C_{5}\}$-free characterization of cochain graphs by Heggernes and Kratsch~\cite{HeggernesK07} implies that
a graph is a cochain graph if and only if it is a cobipartite chordal graph~(see also \cite[Lemma~3.1]{KonagayaOU16}).

\begin{lemma}
\label{lem:pint-diam2->cochain}
If a proper interval graph $G$ has diameter at most~$2$, then $G$ is a cochain graph.
\end{lemma}
\begin{proof}
Fix an interval representation of $G$ where all intervals have length~$1$ and all $2|V(G)|$ endpoints of the intervals are pairwise distinct.
It is known that every proper interval graph admits such a \emph{unit} interval representation~(see \cite{BogartW99}).
Let $I_{\min}$ and $I_{\max}$ be the intervals with 
the minimum left-endpoint $\ell_{\min}$ and the maximum left-endpoint $\ell_{\max}$, respectively,
and $v_{\min}$ and $v_{\max}$ be the corresponding vertices.

We first show that $\ell_{\max} \le \ell_{\min} + 2$.
Suppose otherwise; that is, $\ell_{\max} > \ell_{\min} + 2$.
Then, $v_{\min}$ and $v_{\max}$ are non-adjacent as $I_{\min}$ and $I_{\max}$ are disjoint.
Furthermore, $v_{\min}$ and $v_{\max}$ have no common neighbor:
every interval $I$ intersecting both $I_{\min}$ and $I_{\max}$ has left-endpoint $\ell$ 
with $\ell_{\min} < \ell < \ell_{\min}+1$ and $\ell < \ell_{\max} < \ell+1$;
it follows that $\ell < \ell_{\min}+1 < \ell_{\max}-1 < \ell$, a contradiction.
This implies that the distance between $v_{\min}$ and $v_{\max}$ is more than~$2$ in $G$, contradicting the diameter bound of $G$.

Now we partition the vertices of $G$ into two sets $L$ and $R$:
$L$ is the set of vertices corresponding to the intervals with left-endpoints $\ell$ with $\ell_{\min} \le \ell \le \ell_{\min}+1$;
$R$ is the set of the remaining vertices, which correspond to the intervals with left-endpoints $\ell$ with $\ell_{\min}+1 < \ell \le \ell_{\max} \le \ell_{\min}+2$.
Both $L$ and $R$ form cliques, and hence, $G$ is cobipartite.
Since a proper interval graph is chordal, $G$ is a cochain graph.
\end{proof}

Note that the converse of \cref{lem:pint-diam2->cochain} does not hold in general:
$P_{4}$ is a cochain graph (and thus a proper interval graph) with diameter~$3$.

%%%%%%%%%%%%%%%%%%%%%%%%%%%%%%%%%%%%%%%%%%%%%%%%%%%%%%%%%%%%%%%%%%%%%%%%%%%%%%%%%%%%%%%%%%%%%%%%%%%%%%%%%%%%%%%%%%%%%%%%%%%%%%%%
%%%%%%%%%%%%%%%%%%%%%%%%%%%%%%%%%%%%%%%%%%%%%%%%%%%%%%%%%%%%%%%%%%%%%%%%%%%%%%%%%%%%%%%%%%%%%%%%%%%%%%%%%%%%%%%%%%%%%%%%%%%%%%%%
%%%%%%%%%%%%%%%%%%%%%%%%%%%%%%%%%%%%%%%%%%%%%%%%%%%%%%%%%%%%%%%%%%%%%%%%%%%%%%%%%%%%%%%%%%%%%%%%%%%%%%%%%%%%%%%%%%%%%%%%%%%%%%%%

\section{Concluding remarks}

In this paper, we settled an open case of {\STCfull} by showing that the problem is NP-complete even on proper interval graphs of linear clique-width at most~$4$.
As a byproduct, we also showed that the problem is NP-complete on graphs of diameter~$2$.
These results leave several important questions open.

\paragraph{Graph classes}
From the viewpoint of graph classes, one of the most relevant remaining classes is that of cographs, i.e., graphs of clique-width at most~$2$ (see \cref{fig:graph-classes}).
Note that {\STC} is known to be NP-complete on graphs of clique-width at most~$3$~\cite{OkamotoOUU11}.
Note also that a connected cograph has diameter at most~$2$.

\paragraph{Fixed $k$}
Let $k$ be the target congestion in {\STC}.
If $k \le 3$, then {\STC} is polynomial-time solvable~\cite{BodlaenderFGOL12}.
On the other hand, for every fixed $k \ge 5$, {\STC} is NP-complete~\cite{LuuC25}.
The case of $k = 4$ remains unsettled.

\paragraph{Approximation}
Recently, Kolman~\cite{Kolman25} gave an $O(\Delta \cdot \log^{3/2} n)$-approximation algorithm, where $\Delta$ is the maximum degree.
On the other hand, the NP-hardness of {\STC} with $k = 5$ implies that an approximation factor strictly better than $1.2$ is impossible unless $\mathrm{P} = \mathrm{NP}$.
It would be an interesting challenge to narrow the gap between these upper and lower bounds.

\paragraph{Structural parameterizations}
In a recent paper, Lampis et al.~\cite{LampisMNOVV25} resolved the complexity of {\STC} for many structural parameters, while several cases remain open.
For example, they mention parameterization by neighborhood diversity as an open case.

%%%%%%%%%%%%%%%%%%%%%%%%%%%%%%%%%%%%%%%%%%%%%%%%%%%%%%%%%%%%%%%%%%%%%%%%%%%%%%%%%%%%%%%%%%%%%%%%%%%%%%%%%%%%%%%%%%%%%%%%%%%%%%%%
%%%%%%%%%%%%%%%%%%%%%%%%%%%%%%%%%%%%%%%%%%%%%%%%%%%%%%%%%%%%%%%%%%%%%%%%%%%%%%%%%%%%%%%%%%%%%%%%%%%%%%%%%%%%%%%%%%%%%%%%%%%%%%%%
%%%%%%%%%%%%%%%%%%%%%%%%%%%%%%%%%%%%%%%%%%%%%%%%%%%%%%%%%%%%%%%%%%%%%%%%%%%%%%%%%%%%%%%%%%%%%%%%%%%%%%%%%%%%%%%%%%%%%%%%%%%%%%%%

\bibliographystyle{plainurl}
\bibliography{ref}

\end{document}